\documentclass[10pt,conference]{IEEEtran}
\usepackage{amsfonts}
\usepackage{amsmath}
\usepackage{amssymb}
\usepackage{amsthm}
\usepackage{graphicx}
\usepackage{subcaption}
\usepackage{placeins}
\usepackage{mathdefs}
\usepackage{kbordermatrix}
\usepackage[font={small}]{caption}
\usepackage[font={footnotesize}]{subcaption}

\newtheorem{definition}{Definition}
\newtheorem{example}{Example}
\newtheorem{proposition}{Proposition}
\newtheorem{theorem}{Theorem}
\newtheorem{corollary}{Corollary}

\newcommand{\corref}[1]{Cor.~\ref{cor:#1}}
\newcommand{\defref}[1]{Def.~\ref{def:#1}}

\newcommand{\prpref}[1]{Prop.~\ref{prp:#1}}

\newcommand{\secref}[1]{\S\ref{sec:#1}}
\newcommand{\figref}[1]{Fig.~\ref{fig:#1}}
\newcommand{\tabref}[1]{Tab.~\ref{tab:#1}}

\renewcommand{\eqref}[1]{(\ref{eq:#1})}

\title{A Markov chain model for the search time for max degree nodes in a graph using a  biased random walk}
\author{%
\IEEEauthorblockN{Jonathan Stokes and Steven Weber}%
\IEEEauthorblockA{%
Department of Electrical and Computer Engineering\\
Drexel University, Philadelphia, PA 19104\\
}
\thanks{This research has been supported by the National Science Foundation under award \#IIS-1250786.}
}

\begin{document}
\IEEEoverridecommandlockouts \maketitle 
\begin{abstract} 
We consider the problem of estimating the expected time to find a maximum degree node on a graph using a (parameterized) biased random walk.  For assortative graphs the positive degree correlation serves as a local gradient for which a bias towards selecting higher degree neighbors will on average reduce the search time. Unfortunately, although the expected absorption time on the graph can be written down using the theory of absorbing Markov chains, computing this time is infeasible for large graphs.  With this motivation, we construct an absorbing Markov chain with a state for each degree of the graph, and observe computing the expected absorption time is now computationally feasible.  Our paper finds preliminary results along the following lines: i) there are graphs for which the proposed Markov model does and graphs for which the model does not capture the absorbtion time, ii) there are graphs where random sampling outperforms biased random walks, and graphs where biased random walks are superior, and iii) the optimal bias parameter for the random walk is graph dependent, and we study the dependence on the graph assortativity.
\end{abstract} 

\section{Introduction}

A graph representing Facebook's network of 1.4 billion users would require 1.4 billion nodes and hundreds of billions of edges, stretching the capacity of current hardware to hold the graph in memory.  This inability to represent large graphs makes them difficult to study. One way of sidestepping this issue is to study a representative subsample of the entire graph.  How one takes this sample often depends on the properties of the graph being studied. The simplest method of sampling a graph is to select nodes uniformly at random, however, many graphs, specifically those representing social networks, exhibit a power law degree distributions the probability of selecting a high degree node using this method is small. Intuitively a better sampling method would use the information gained by a sample to increase the probability of selecting a max degree node on subsequent samples. In the context of social networks, sampling methods for a max degree node that exploit local information exploit the friendship paradox; on average your friends have more friends than you do. {\em One goal is to study the impact that biasing the random choice of the next neighbor in a random walk towards selecting higher degree neighbors has on the time to reach a maximum degree node.}

Previous work has developed analytical bounds for the hitting time of a biased random walk, the time it takes to get from one node in a graph to another, and the cover time of the walk, the time it takes a walk to visit every node in a graph. Ikeda shows that the hitting and cover time of an undirected graph of $n$ nodes is upper bounded by $O(n^2)$ and $O(n^2 log(n))$ respectively \cite{IkeKub2003}.  Cooper shows all nodes of degree $n^a$ for $a < 1$ in a $n$ node power law graphs with exponent $c$ can be found in $O(n^{1-a(c-2)+\delta})$ steps \cite{CooRad2014}. Maiya evaluates several biased sampling strategies on real world graphs numerically showing that a walk which always transitions to the max degree node attached to its current node is a good method of exploring these graphs \cite{MaiWol2011}.

Our approach is different.  We first observe that the basic theory of absorbing finite-state Markov chains (\secref{KemSne}) yields expressions for the mean and variance of the random time to reach an absorbing state, which in our case is a maximum degree node on the $n$-node graph.  However, {\em computing} this mean and variance requires inverting an $n \times n$ matrix, which for large graphs is prohibitive (\secref{gm}).  Consequently, we consider the case of assortative graphs where the degrees of the endpoints of an edge are positively correlated, and recognize that for such graphs an intelligent strategy for minimizing the search time to find a maximum degree node is to exploit the local gradient provided by the assortativity, and to select the next node in the walk with a bias towards higher degree neighbors.  We construct a random walk on a significantly reduced state space, with one state for each degree in the graph, and the transition probability matrix derived from the joint degree distribution of the graph and the random walk bias parameter (\secref{approx}). The advantage of such a model is that we can analytically compute the absorption time of the random walk on this reduced state space quite easily. {\em A second goal is to study how the statistics of the random absorption time to find maximum degree nodes on large graphs may be captured using this model.}

As shown in \secref{results}, we investigate three natural questions, for which this paper offers only preliminary and numerical results: $i)$ for which graphs does the above Markov chain state reduction accurately capture the mean absorbtion time?, $ii)$ for which graphs does a biased random walk find a maximum degree node more quickly than does random sampling?, and $iii)$ how does the optimal bias parameter ($\beta$) depend upon the graph?  For our preliminary numerical investigation of these three questions, we employ Erd\H{o}s-R\'{e}nyi (ER) graphs, rewired (using a standard rewiring algorithm) to ensure the graph is connected and has a target assortativity ($\alpha \in [-1,+1]$); the assortativity plays the role of an independent control parameter in our investigations.  Our preliminary results suggest the following answers to the above questions.  First, there are certain $(\alpha,\beta)$ pairs for which the above model does and does not work well, in particular $\alpha < 0$ and $\beta$ large gives a poor match.  Second, for certain values of $\alpha$ the optimal $\beta^*(\alpha)$ are such that the resulting optimal absorbtion time of a biased random walk is superior to random sampling, but otherwise not.  Third, the optimal $\beta^*(\alpha)$ appears to have an increasing trend in $\alpha$, so that ``following the local gradient (choosing a maximum degree neighbor)'' is optimal for highly assortative graphs, while choosing a neighbor uniformly is superior for highly disassortative graphs.

\section{Absorption time for Markov chains}
\label{sec:KemSne}

We will have cause to use the theorem and corollary below on the mean and variance of the random absorption time $T_X$ of an absorbing finite-state discrete-time Markov chain (DTMC) $X = (X(t), t \in \Nbb)$ taking values in a finite state space $\Xmc$ (with $|\Xmc|=n$). Partition $\Xmc$ into absorbing states $\hat{\Xmc}$ (with $|\hat{\Xmc}|=\hat{n}$) and transient states $\check{\Xmc}$ (with $|\check{\Xmc}|=\check{n}=n-\hat{n}$), and so too partition the $n \times n$ transition matrix $P_X$ into submatrices 
\begin{equation}
P_X =\kbordermatrix{
             & \hat{\Xmc} & \check{\Xmc} \\
\hat{\Xmc}   & A          & O  \\
\check{\Xmc} & R          & Q
},
\end{equation}
where $A$ is the $\hat{n} \times \hat{n}$ submatrix of transition probabilities between absorbing states, $O$ is the $\hat{n} \times \check{n}$ zero matrix, $R$ is the $\check{n} \times \hat{n}$ submatrix of transition probabilities from transient to absorbing states, and $Q$ is the $\check{n} \times \check{n}$ submatrix of transition probabilities between transient states.  

\begin{definition}
\label{def:fundmatabtime}
The {\em fundamental matrix} for the DTMC $X$ is the $\check{n}\times\check{n}$ matrix $N_X = (I_{\check{n}}-Q)^{-1}$, for $I_{\check{n}}$ the $\check{n} \times \check{n}$ identity  matrix.  The {\em absorption times} $\Tbf_X = (T_X(x), x \in \check{\Xmc})$ are the collection of $\check{n}$ random absorption times starting from each possible initial transient state, with $T_X(x) = \min\{t \in \Nbb : X(t) \in \hat{\Xmc}|X(0) = x)$.
\end{definition}

\begin{theorem} (\cite{KemSne1983})
\label{thm:KemSne}
The absorption times $\Tbf_X$ have means $\boldsymbol\mu_X$ and variances $\boldsymbol\sigma^2_X$:
\begin{equation}
\boldsymbol\mu_X = N_X \mathbf{1}_{\check{n}}, ~
\boldsymbol\sigma^2_X = (2 N_X - I_{\check{n}}) \boldsymbol\mu_X - (\boldsymbol\mu^{\Tsf}_X \boldsymbol\mu_X) \mathbf{1}_{\check{n}},
\end{equation}
where $\mathbf{1}_{\check{n}}$ is the $\check{n}$-vector of ones.
\end{theorem}

Let $\check{\pbf}_X = (\check{p}_X(x), x \in \check{\Xmc})$ be the initial distribution of $X$ on the transient states $\check{\Xmc}$, i.e., $X(0) \sim \check{\pbf}_X$, and define $T_X$ as the corresponding random absorption time.

\begin{corollary}
\label{cor:KemSne}
Given the initial distribution $\check{\pbf}_X$ for $X(0)$ on $\check{\Xmc}$, the resulting absorption time $T_X$ has
\begin{equation}
\label{eq:etxvartx}
\Ebb[T_X] = \sum_{x \in \check{\Xmc}} \check{p}_X(x) \mu_X(x), ~
\mathrm{Var}(T_X) = \sum_{x \in \check{\Xmc}} \check{p}_X(x) \sigma^2_X(x).
\end{equation}
\end{corollary}

\section{Biased random walk on the graph}
\label{sec:gm}

Denote the undirected graph on which the search takes place as $\Gmc = (\Vmc,\Emc)$, with $\Vmc$ the set of vertices and $\Emc$ the set of undirected edges.  The order and size of $\Gmc$ are $|\Vmc|=n$ and $|\Emc|=m$, respectively.  The neighborhood of any node $v$ is denoted $\Gamma(v) = \{ u \in \Vmc : \{u,v\} \in \Emc\}$, and the degree is $d(v) = |\Gamma(v)|$.  The set of degrees found in the graph is $\Dmc = \bigcup_{v \in \Vmc} d(v)$, the number of distinct degrees is $\delta = |\Dmc|$.  The nodes $\Vmc$ are partitioned into subsets $(\Vmc_k, k \in \Dmc)$ with $\Vmc_k = \{ v \in \Vmc : d(v) = k\}$ the set of nodes of degree $k$, and the resulting degree distribution is $\pbf_{\Dmc} = (p_{\Dmc}(k), k \in \Dmc)$, with entries $p_{\Dmc}(k) = \frac{|\Vmc_k|}{n}$.  In particular, the max degree found in the graph is $d_{\rm max} = \max \Dmc$, and the set of nodes with max degree is $\Vmc_{\rm max} = \Vmc_{d_{\rm max}}$.  An absorbing biased random walk on $\Gmc$ is defined as follows, where the notation is adopted from the general case presented in \secref{KemSne}.

\begin{definition}
\label{def:dtmcgraph}
The absorbing discrete-time Markov chain (DTMC) $V = (V(t), t \in \Nbb)$ on $\Gmc$ with bias parameter $\beta \geq 0$ has states $\Vmc$, partitioned into absorbing states $\hat{\Vmc} = \Vmc_{\rm max}$ and transient states $\check{\Vmc} = \Vmc \setminus \Vmc_{\rm max}$. The $n \times n$ transition probability matrix $P_V$ has entries $P_V(u,v) = \Pbb(V(t+1) = v|V(t)=u)$.  For $u \in \check{\Vmc}$ we set $P_V(u,v) = \frac{d(v)^{\beta}}{\sum_{v' \in \Gamma(u)} d(v')^{\beta}}$ for $v \in \Gamma(u)$ and $P_V(u,v) = 0$ else.  For $u \in \hat{\Vmc}$: $P_V(u,u) = 1$ and $P_V(u,v) = 0$ for $v \neq u$.  The absorption times from each initial transient state are $\Tbf_V = (T_V(v), v \in \check{\Vmc})$, with $T_V(v) = \min\{t \in \Nbb : V(t) \in \hat{\Vmc} | V(0) = v\}$.  Note the event $V(t) \in \hat{\Vmc}$ is equivalent to $d(V(t)) = d_{\rm max}$.
\end{definition}

This random search for $\Vmc_{\rm max}$ uses a transition probability that is biased towards selecting higher degree neighbors of the current node, as introduced by Cooper \cite{CooRad2014}, with the bias monotonically increasing in $\beta$.  For $\beta = 0$ the next step is uniform among all neighbors of $u$, while as $\beta \to \infty$ the next step is uniform among the maximum degree neighbors of $u$, i.e., $\Gamma_{\rm max}(u) = \argmax_{v \in \Gamma(u)} d(v)$.  A key goal of this paper is to characterize $\Ebb[T_V]$ and $\mathrm{Var}(T_V)$ as a function of the initial distribution $\check{p}_V$ on $\check{\Vmc}$ and on the bias parameter $\beta$.  Although numerical estimates can (and will) be obtained by simulating the random walk, for large graphs (large $n$) it is computationally infeasible to compute $\Ebb[T_V]$ and $\mathrm{Var}(T_V)$ analytically using \corref{KemSne}, due to the need to compute $N_V$ as the inverse of the (large) $\check{n} \times \check{n}$ matrix $I_{\check{n}}-Q$.  This difficulty motivates us to define an approximation of the biased random walk using a significantly smaller state space, discussed next.

\section{Approximate biased random walk}
\label{sec:approx}

In this section we develop our key approximation giving a more computationally feasible means to estimate the mean and standard deviation of $T_V$.  Our derivation consists of the following steps: $i)$ define the joint degree matrix $J$ and the conditional degree distribution matrix $\tilde{J}$ for $\Gmc$, $ii)$ define the biased walk degree transition probability distribution $\breve{\pbf}(\nbf)$ and  average biased walk degree transition matrix $\bar{P}$, and $iii)$ define the approximate biased walk degree transition matrix $\tilde{P}$.  

\subsection{Joint degree and conditional degree distribution matrices}

We define the joint degree matrix $J$ and the conditional degree distribution matrix $\tilde{J}$ for $\Gmc$.

\begin{definition} 
\label{def:jointdegreematrix} 
The joint degree matrix $J$ is the $\delta \times \delta$ symmetric matrix with entries $J(k,l) = \#\{\{u,v\} \in \Emc : \{d(u),d(v)\} = \{k,l\}\}$, the number of edges in $\Gmc$ with endpoints of degrees $k$ and $l$, and entries $J(k,k)$ twice the number of edges with endpoints both of degree $k$, for $(k,l) \in \Dmc^2$.
\end{definition}

\begin{example} 
\label{exa:jointdegreematrix} 
The joint degree matrix $J$ for the graph $\Gmc$ shown in \figref{jointdegreematrix} (left) is found by grouping the edges $\Emc$ by the degrees of the endpoints, as shown in \figref{jointdegreematrix} (right): 
\begin{equation}
\label{eq:jointdegreematrix}
J =\kbordermatrix{
  & 1 & 2 & 3 & 4 \\
1 &   &   &   & 1 \\
2 &   &   & 1 & 3 \\
3 &   & 1 &   & 2 \\
4 & 1 & 3 & 2 & 2 
}.
\end{equation}
\end{example}

\begin{figure}[!ht]
\centering
\includegraphics[width=\columnwidth]{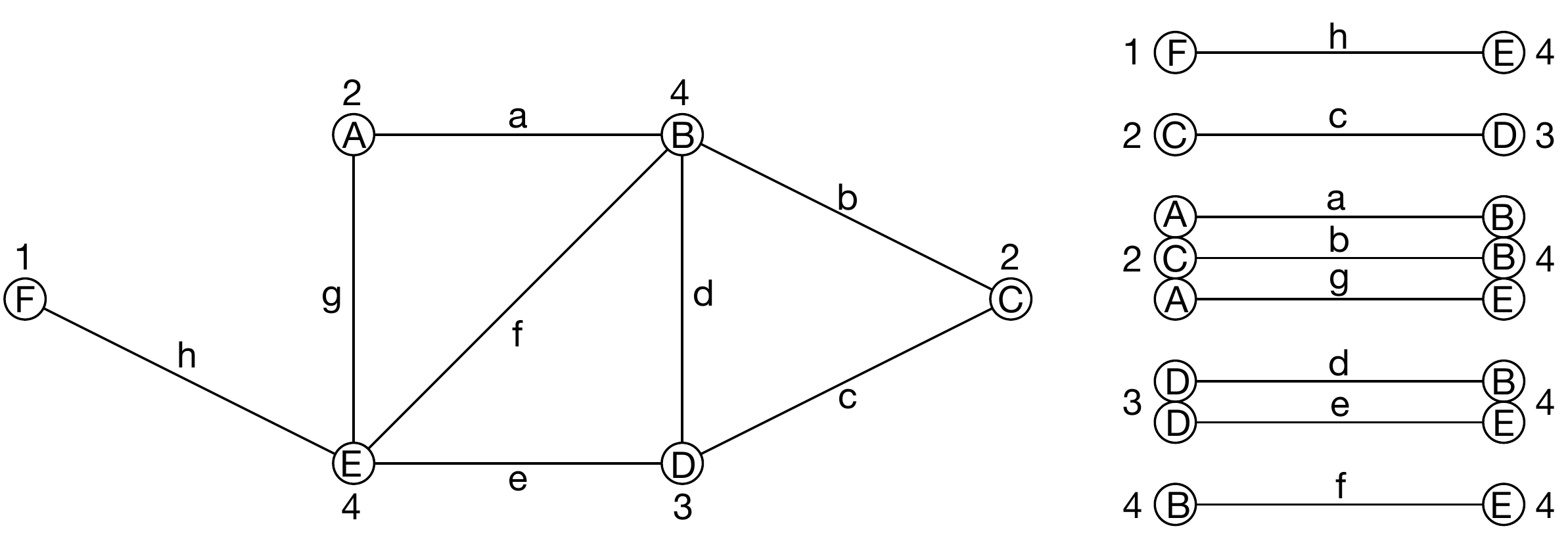}
\caption{Graph $\Gmc$ (left) with edges $\Emc$ grouped by the endpoint degrees (right).}
\label{fig:jointdegreematrix}
\end{figure}

Observe $\sum_{k,l} J(k,l) = 2 m$, and that if an edge from $\Emc$ is selected uniformly at random then the probability that it should have endpoints $\{k,l\}$ is given by $J(k,l)/m$ for $k \neq l$, and $J(k,k)/(2m)$ else.  The conditional degree distribution matrix is obtained by normalizing each row of $J$ into a probability distribution.

\begin{definition}
\label{def:conddegreematrix}
The conditional degree distribution matrix $\tilde{J}$ is the $\delta \times \delta$ matrix with $\tilde{J}(k,l) = J(k,l)/\sum_{l' \in \Dmc} J(k,l')$, for $(k,l) \in \Dmc^2$.
\end{definition}

It is important to motivate $\tilde{J}$ for what follows.  Towards this end we define the following.  First, the set of degrees $\Dmc(v) \subseteq \Dmc$ of some $v \in \Vmc$ is formed by taking the union of the degrees of each neighbor of $v$, i.e., $\Dmc(v) = \bigcup_{u \in \Gamma(v)} d(u)$.  The number of neighbors of some $v \in \Vmc$ of each different degree $\Dmc(v)$, which we term the {\em degree neighborhood of $v$}, is given by $\nbf(v) = (n_l(v), l \in \Dmc(v))$, with entries $n_l(v) = \#\{u \in \Gamma(v) : d(u) = l\}$.  

\begin{proposition}
\label{prp:conddegdistave}
The entries $\tilde{J}(k,l)$ of the conditional degree distribution matrix in \defref{conddegreematrix} are the averages of the fraction of neighbors of degree $l$ over all degree $k$ nodes:
\begin{equation}
\tilde{J}(k,l) = \frac{1}{|\Vmc_k|} \sum_{v \in \Vmc_k} \frac{n_l(v)}{k}.
\end{equation}
\end{proposition}

\begin{proof} 
Observe $J(k,l) = \sum_{v \in \Vmc_k} n_l(v)$ by partitioning all edges with a degree $k$ endpoint by the degree of the other endpoint, and $\sum_{l' \in \Dmc} J(k,l') = k |\Vmc_k|$ since the sum is the number of edges with a degree $k$ endpoint.  
\end{proof}

Row $k$ of $\tilde{J}$ is denoted $\tilde{J}(k,\cdot) = (\tilde{J}(k,l), l \in \Dmc)$.  The locations of the non-zero entries of $\tilde{J}(k,\cdot)$ are denoted $\Dmc_k = \bigcup_{v \in \Vmc_k} \Dmc(v)$, i.e., $l \in \Dmc_k$ means there exists some $\{u,v\} \in \Emc$ with $d(u)=k$ and $d(v) = l$.  

\subsection{Biased walk degree transition probability distribution}

We define the biased walk degree transition probability distribution $\breve{\pbf}(\nbf)$ and average biased walk degree transition matrix $\bar{P}$.  

\begin{definition} 
The biased random walk degree transition probability distribution $\breve{\pbf}(\nbf) = (\breve{p}_l(\nbf), l \in \Dmc)$ from a node with degree neighborhood $\nbf = (n_l, l \in \Dmc)$ (with $\sum_{l \in \Dmc} n_l \in \Dmc$) is 
\begin{equation}
\breve{p}_l(\nbf) = \frac{n_l l^{\beta}}{\sum_{l' \in \Dmc} n_{l'} l'^{\beta}}.
\end{equation}
\end{definition}

The following is an immediate result of \defref{dtmcgraph}.

\begin{proposition}
\label{prop:degtransprob}
The biased random walk $V$ in \defref{dtmcgraph} obeys the degree transition probability 
\begin{equation}
\Pbb(d(V(t+1))=l|V(t)=v) = \breve{p}_l(\nbf(v)), ~ l \in \Dmc(v).
\end{equation}
\end{proposition}

The above proposition states that the biased random walk $V$ degree transition probability distribution has the property that the probability of the degree of the next node of the biased random walk depends upon the current node $v$ only through $\nbf(v)$.  The next definition gives an average transition probability from nodes of degree $k$ to nodes of degree $l$ under the biased random walk.

\begin{definition}
The average biased walk degree transition matrix $\bar{P}_V$ is the $\delta \times \delta$ matrix with entries 
\begin{equation}
\bar{P}_V(k,l) = \frac{1}{|\Vmc_k|} \sum_{v \in \Vmc_k} \breve{p}_l(\nbf(v)), ~ (k,l) \in \Dmc^2
\end{equation}
giving the average probability of transitioning to a node of degree $l$ over all starting nodes of degree $k$.
\end{definition}

A key point, developed below, is that $\bar{P}_V$ is of size $\delta \times \delta$ whereas the transition matrix $P_V$ is of (potentially) significantly larger size $n \times n$.  

\subsection{Approximate biased walk degree transition matrix}

A random $\delta$-vector $\Nbf = (N_l, l \in \Dmc)$ has a {\em multinomial distribution} with parameters $(k,\pbf)$, denoted $\Nbf \sim \mathrm{mult}(k,\pbf)$, if 
\begin{equation}
\Pbb(\Nbf = \nbf) = \binom{k}{\nbf} \prod_{l \in \Dmc} p_l^{n_l}, ~ \nbf \in \Nmc_k,
\end{equation}
where $\binom{k}{\nbf} = \binom{k}{\prod_{l \in \Dmc} n_l!}$ is the multinomial coefficient. Here, $\Nbf$ has support $\Nmc_k = \{ \nbf \in \Nbb^{\delta} : \sum_{l \in \Dmc} n_l = k\}$, defined as the set of all possible $\delta$-vectors from $\Nbb^{\delta}$ that sum to $k$.  

The approximate biased walk degree transition matrix $\tilde{P}_W$ has entries defined using the expectation of the biased random walk degree transition probability distribution $\breve{\pbf}_l$ when the degree neighborhood $\nbf$ is taken as a random multinomial vector $\Nbf$ with parameters $(k,\tilde{J}(k,\cdot))$.  

\begin{definition}
\label{def:abwdtm}
The approximate biased walk degree transition matrix $\tilde{P}_W$ is the $\delta \times \delta$ matrix with entries 
\begin{equation}
\tilde{P}_W(k,l) = \Ebb \left[ \breve{p}_l(\Nbf) \right], ~ \Nbf \sim \mathrm{mult}(k,\tilde{J}(k,\cdot)).
\end{equation}
That is:
\begin{equation}
\label{eq:approxtransmatrix}
\tilde{P}_W(k,l) = \sum_{\nbf \in \Nmc_k} \binom{k}{\nbf} \prod_{l \in \Dmc_k} \tilde{J}(k,l)^{n_l} \frac{n_l l^{\beta}}{\sum_{l' \in \Dmc_k} n_{l'} l'^{\beta}}.
\end{equation}
\end{definition}

Our model is suitable for those graphs for which $\bar{P}_V(k,l) \approx \tilde{P}_W(k,l)$ for each $(k,l) \in \Dmc^2$; as will be shown in \secref{results}, it is not difficult to identify graphs for which the model works, and to find graphs for which it does not.  Finally, we define a biased random walk on $\Dmc$.

\begin{definition}
\label{def:dtmcdeg}
The absorbing discrete-time Markov chain (DTMC) $W = (W(t), t \in \Nbb)$ on $\Dmc$ with bias parameter $\beta \geq 0$ has states $\Dmc$, partitioned into absorbing states $\hat{\Dmc} = \{d_{\rm max}\}$ and transient states $\check{\Dmc} = \Dmc \setminus d_{\rm max}$.  The $\delta \times \delta$ transition probability matrix $\tilde{P}_W$ is given in \defref{abwdtm}.  The absorption times from each initial transient state are $\Tbf_W = (T_W(l), l \in \check{\Dmc})$, with $T_W(l) = \min\{t \in \Nbb : W(t) = d_{\rm max} | W(0) = l\}$.  
\end{definition}

Recall that for large graphs (large $n$) it is not possible to compute $\Ebb[T_V]$ and $\mathrm{Var}(T_V)$ for $T_V$ in \defref{dtmcgraph} using \corref{KemSne} since we cannot obtain the fundamental matrix $N_V = (I_{\check{n}} - Q)^{-1}$, on account of the difficulty of inverting the (large) $\check{n} \times \check{n}$ matrix $I_{\check{n}}-Q$.  In contrast, for graphs with bounded $\delta$, it {\em is} possible to compute $\Ebb[T_W]$ and $\mathrm{Var}(T_W)$ for $T_W$ in \defref{dtmcdeg} using \corref{KemSne} since the corresponding fundamental matrix $N_W = (I_{\delta-1}-Q)^{-1}$ is obtained by inverting the (smaller) $(\delta-1)\times(\delta-1)$ matrix $I_{\delta-1}-Q$. 

Our model's approximation $T_V \approx T_W$ lies with the required approximation $\bar{P}_V \approx \tilde{P}_W$ mentioned above.  That is, $i)$ although the biased random walk on the graph has a probability of transitioning to nodes of each degree $l \in \Dmc(v)$ from a node $v$ of degree $k$ that depends upon the exact degree neighborhood $\nbf(v)$ (i.e., $\breve{p}_l(\nbf(v))$), $ii)$ we approximate this probability using the expectation with respect to a {\em random} degree neighborhood $\Nbf$, drawn with parameters $k$ and $\tilde{J}(k,\cdot)$. This distribution $\tilde{J}(k,\cdot)$ is the average distribution of the number of nodes of each degree over all degree $k$ nodes (\prpref{conddegdistave}). 

\section{Results}
\label{sec:results}

The experimentation framework for this work is written using the igraph Python
library \cite{CsaNep2006}.  We use this framework to generate instances of a
graph family, then for each graph we measure the random times to find a maximum degree node by $i)$ a biased random walk (BRW), and $ii)$ random sampling.  Unless noted otherwise, the mean $\Ebb[T_V]$ and standard deviation $\mathrm{Std}[T_V]$ of the absorption of a BRW on a graph is calculated from $500$ trials for each tested bias coefficient $\beta$.  We compare the empirical mean $\Ebb[T_V]$ and empirical standard deviation $\mathrm{Std}[T_V]$ measured from the BRW on the graph with the numerical mean $\Ebb[T_{W}]$ and standard deviation $\mathrm{Std}[T_{W}]$ of a BRW on the Markov chain of the graph's degree states.  \tabref{g_param} describes the Erd\H{o}s-R\'{e}nyi (ER) graphs we used.
\begin{table}[h]
\centering
\begin{tabular}{|c|c|c|c|}
\hline
Graph 					& Parameters 	& Size	& $d_{\rm max}$ \\ 	\hline	
Erd\H{o}s-R\'{e}nyi (ER)& $p=0.05$		& 100	& 11 \\	\hline	
Erd\H{o}s-R\'{e}nyi (ER)& $p=0.0024$	& 1121	& 10 \\	
Erd\H{o}s-R\'{e}nyi (ER)& $p=0.002569$	& 1011	& 11 \\	
\hline
\end{tabular}
\caption{Parameters of the graphs used in the simulations.}
\label{tab:g_param}
\end{table}

As currently implemented our model has a significant computational limitation in that the approximate random walk $(W(t))$ in \defref{dtmcdeg} requires computing the matrix $\tilde{P}_W$ in \eqref{approxtransmatrix}, and each such entry $\tilde{P}_W(k,l)$ requires summing over all $\nbf \in \Nmc_k$.  As the size of $|\Nmc_k|$ grows exponentially in $k$, we are unable to compute it for graphs with $d_{\rm max} > 15$; this is the motivation behind our selecting $(n,p)$ pairs for the ER graphs so that $d_{\rm max}$ is small.  Addressing this deficiency is the subject of our ongoing and future work.

\subsection{Erd\H{o}s-R\'{e}nyi (ER) Graphs}

The Erd\H{o}s-R\'{e}nyi graph \cite{ErdRen1959} $G(n,p)$ is a family of random graphs with $n$ nodes and each of the $\binom{n}{2}$ possible edges is added independently with probability $p$.  The resulting graph has a binomial degree distribution, $p_{\Dmc}(k) \sim \mathrm{bin}(n-1,p)$.  To investigate larger graphs we must ensure the graph is constructed so as to have a bounded expected maximum degree, on account of the computational limitation $d_{\rm max} \leq 15$ discussed above.  The following proposition gives an upper bound on the expected maximum of $n$ iid random variables:
\begin{proposition}
\label{prp:expectedmax}
Let $(Y_1,\ldots,Y_n)$ be iid with moment generating function (MGF) $\phi(t) = \Ebb[\erm^{tY}]$ and let $Y_{\rm max} = \max (X_1,\ldots,X_n)$.  Then 
$\Ebb[Y_{\rm max}] \leq \frac{1}{t}\log(n\phi(t))$.
\end{proposition}
It may be proved by application of Jensen's inequality to establish $\erm^{t \Ebb[Y_{\rm max}]} \leq n \phi(t)$.

We apply the above rule to the random degrees $(D(v), v \in [n])$ of an ER graph $G(n,p)$, each $D(v) \sim \mathrm{bin}(n-1,p)$, which are identically distributed, but not independent; it can be shown that the slight dependence is inessential and the bound applies.  Recall that the binomial distribution $\mathrm{bin}(n,\lambda/n)$ can be approximated as a Poisson distribution $\mathrm{Po}(\lambda)$ in the case when $n$ is large, and that the Poisson MGF is $\phi(t) = \erm^{\lambda(\erm^t-1)}$.  Applying \prpref{expectedmax} to this case yields the following upper bound on the max degree of an ER graph $G(n,\lambda/n)$ 
\begin{equation}
\Ebb[\max(D(v), v \in [n])] \leq \frac{\log(n) - \lambda}{\Wmc (\frac{\log(n)-\lambda}{\erm \lambda})}
\label{eq:up_bound_2}
\end{equation}
where $\Wmc$ is the Lambert $W$ function.  The value of \eqref{up_bound_2} is that it allows us to select $\lambda(n)$ so that the resulting ER graph of order $n$ has a specified expected max degree upper bound.  

However, the ER graph is known to be disconnected with high probability when $p(n) = \lambda/n$ (or smaller) for any $\lambda$, and our random walk is only guaranteed to find a maximum degree node for a connected graph. It is further known that an ER graph with $p(n) = \lambda(\beta)/n$ will have a fraction $\beta$ of the $n$ nodes in the giant connected component where $\lambda(\beta) = -\log(1-\beta)/\beta$ \cite{JanLuc2000}.  Hence the trade-off we face in generating ER graphs is between a large fraction $\beta$ in the giant connected component vs.\  a bounded max degree.  For $n=1090$ and $\lambda = 2.8$ (i.e., $p(n) = \lambda/n = 0.0024569)$, the expected max degree upper bound in \eqref{up_bound_2} is $d_{\rm max} \leq 11.1041$, and $\beta = 0.924975$ obeys $-\log(1-\beta)/\beta = \lambda$, meaning the giant connected component will contain approximately $n \beta = 1008$ nodes.  


\subsection{Results for ER Graphs}


\figref{er1200rvsr} (left) shows the expected absorbtion time $\Ebb[T]$ for an ER graph using both $i)$ the simulated biased random walk (BRW) on the original graph $\Ebb[T_V]$, and $ii)$ the analytically computed $\Ebb[T_W]$ using the reduced state space model, both swept over a range of bias coefficients $\beta$.  The plot shows a significant deviation between the measured quantity and the model prediction.  The failure of the model for this graph may be accounted for by the fact that the ER graph is known to have zero assortativity, i.e., the degrees of the two endpoints of the graph are conditionally independent, and as such the degree of the current node does not provide any substantial information about the proximity of that node to higher degree (and by extension, maximum degree) nodes.  In this sense, our approximation is shown to break down for such graphs, as a central assumption of our model is the idea that the degree of a node {\em does} contain information about its degree neighborhood and all nodes of degree $k$ have similar degree neighborhoods. 

\begin{figure}
\centering
\includegraphics[width=1.7in]{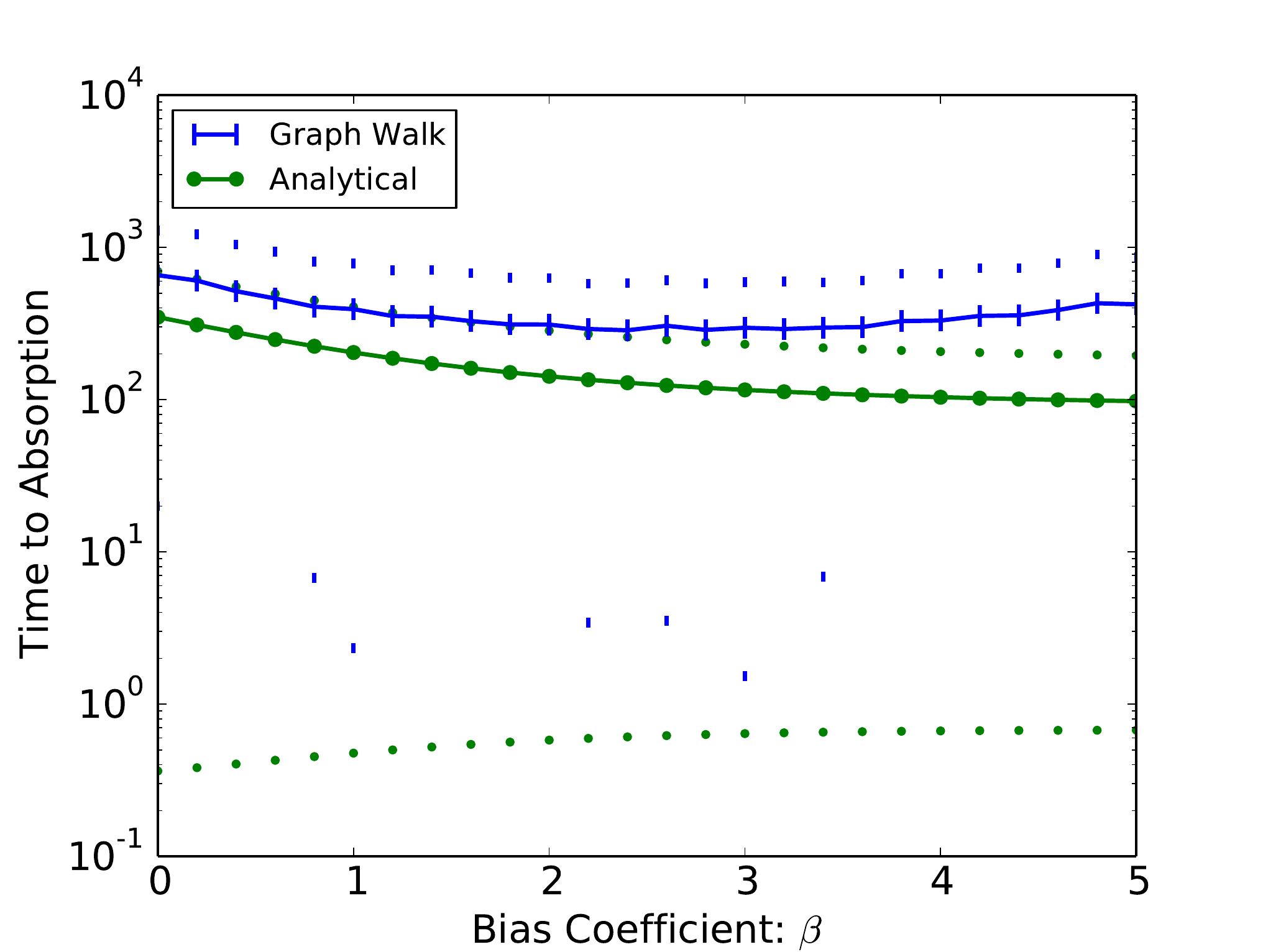}    
\includegraphics[width=1.7in]{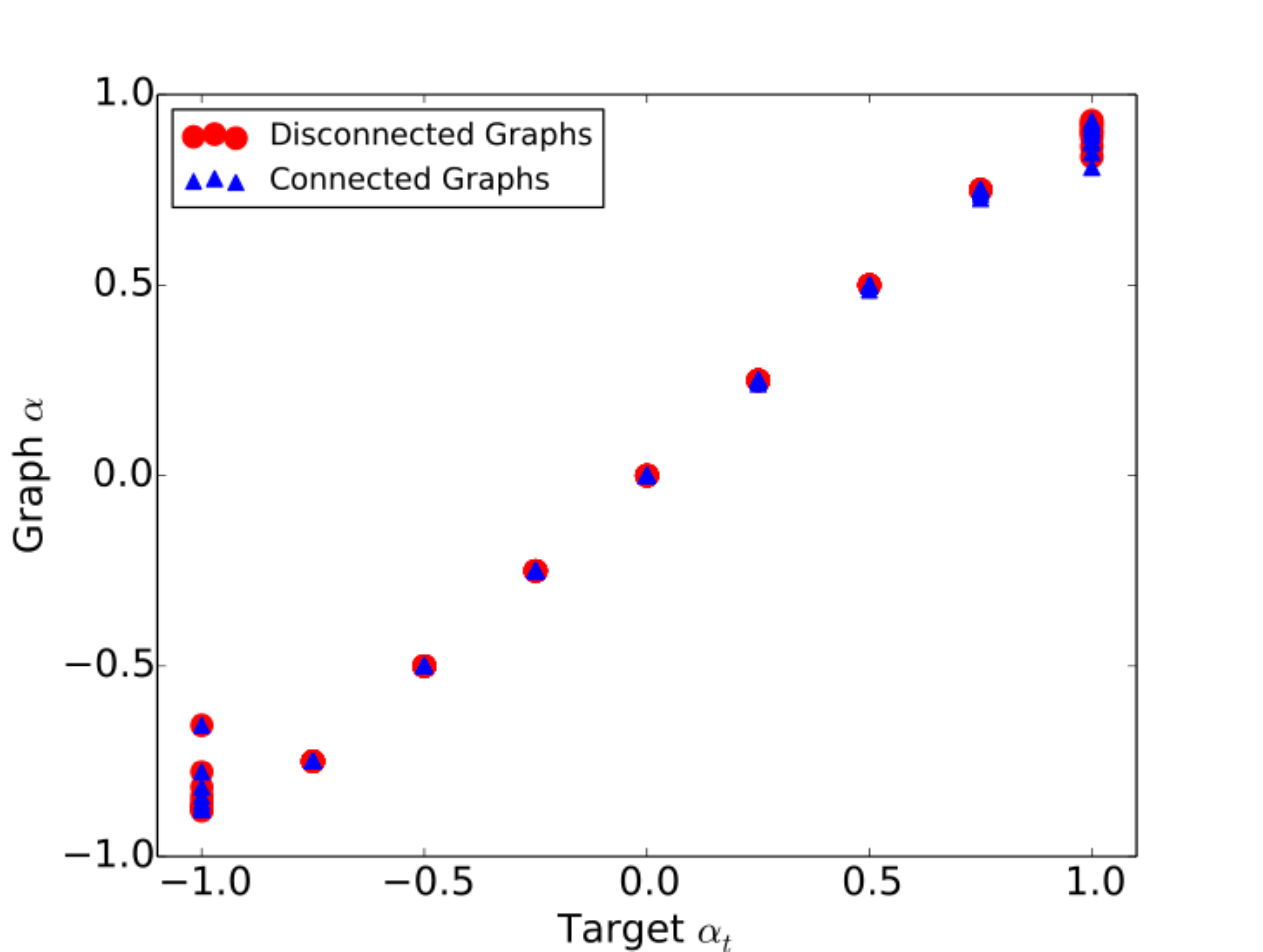}    
\caption{{\bf Left:} average absorbtion time, $\Ebb[T]$ (solid lines), for original graph (blue, via simulation) and model (green, via \eqref{etxvartx}), with $\Ebb[T]\pm\mathrm{Std}[T]$ (dashed lines).  {\bf Right:} output assortativity $\alpha$ as a function of input target assortativity $\alpha_t$ in the random rewiring algorithm.}
\label{fig:er1200rvsr}
\end{figure}

This led to a set of questions: $i)$ is the accuracy of our model of a BRW on a graph dependent upon the graph's assortativity?, $ii)$ are there graphs where BRW finds the max degree nodes faster than random sampling?, and $iii)$ what is the optimal bias coefficient $\beta^*$ for a graph of assortativity $\alpha$? 

\subsection{Graph Re-wiring Algorithm}

Our approach in this paper is to offer some preliminary numerical answers to these questions, using graph assortativity, denoted by $\alpha \in [-1,+1]$, as the independent control parameter.  To construct graphs with a target assortativity, $\alpha_t$, we modified Brunet's rewiring algorithm for increasing or decreasing a graph's assortativity \cite{Xul2005}.  Given an initial graph $G_0$ and a target assortativity $\alpha_t$. We calculate $G_0$'s assortativity, $\alpha_0$ and choose two edges at random $e_1,e_2$. Then we remove $e_1$ and $e_2$ from $G_0$. If two new edges $e_3$ and $e_4$ can be wired between the endpoints of the former edges $e_1$ and $e_2$ without creating self loops or multiple edges such that the assortativity of the new graph $\alpha_1$ is closer to $\alpha_t$ than $\alpha_0$ we add edges $e_3$ and $e_4$, if not we replace $e_1$ and $e_2$. This procedure is repeated until the graph's assortativity is within a suitably small interval around $\alpha_t$. Notice that this procedure preserves the degree distribution of $G_0$, since the degree of the end points of $e_1$ and $e_2$ are unchanged. When the assortativity converges, if the graph is disconnected, then for each disconnected component a random node is selected in the graph's giant component and wired to the smaller component, thereby connecting the graph. The assortativity of the $1011$ node graphs compared to their target assortativity is shown in \figref{er1200rvsr} (right) for both disconnected (red) and connected (blue) graphs. From these results we infer that connecting the graphs in this manner has little effect on their assortativity.

\subsection{Biased Random Walks} 

We carried out Monte carlo simulations of $500$ trials on $10$ graphs with binomial degree distributions of $100$ and $1011$ nodes, while sweeping the bias coefficient $\beta$ of the walk. We compared BRWs with two random sampling algorithms: $i)$ sampling nodes without replacement, denoted 'no-r', and $ii)$ sampling a node and all of its neighbors without replacement, denoted 'no-r n'.  The rationale for these two forms of sampling is in the interest in making a fair comparison in the absorbtion time between the BRW and a random sample.  The BRW algorithm presumes at each step that the search is able to not only view the degree of the current node but also the degree of all neighbors of that node.  Thus any comparison between the performance of, say, $k$ steps of the BRW and $k$ nodes sampled without replacement is unfair to random sampling, since the latter does not see as many nodes as the former.  The second sampling scheme, where at each step we sample a node and see its degree as well as the degrees of its neighbors, offers a more balanced comparison with the performance of the BRW.

For our comparison between the BRW and random sampling we investigated nine sample values of the target assortativity $\alpha_t$, namely the nine values $\{-1.0,-0.75,\ldots,0.75,1.0\}$, and we used both $100$ node and $1000$ node ER graphs.  For each target $\alpha_t$, and each graph rewired to that $\alpha_t$, we swept the biased coefficient $\beta$ over the range $[0,8]$.  

The results for $\alpha_t \in \{+0.5,-0.5\}$ are shown in \figref{ws1000assdis} (for $n=1000$) and \figref{wms100assdis} (for $n=100$).  Several points bear mention.  First, for target assortativity $\alpha_t = +0.5$ (both for $n \approx 1000$ and $n \approx 100$) there exist optimized $\beta^*(\alpha_t)$ for which BRW outperforms random sampling of a node and its neighbor degrees, and there exist non-optimal values of $\beta$ for which random sampling outperforms the non-optimized BRW.  Second, for target assortativity $\alpha_t = -0.5$ we see that the BRW is inferior to random sampling a node and its neighbor degrees for {\em all} values of $\beta$.  Third, the case of $n \approx 1000$ and $\alpha = +0.5$ shows there exists a non-trivial value of $\beta^*$ for this $\alpha$, meaning the optimal value is not at either endpoint of the $\beta$ interval of $[0,8]$.  This suggests that although BRWs can yield superior search times compared with sampling neighborhoods, doing so requires a correctly-tuned value of $\beta$ for the particular value of $\alpha$ (and in this case also $n$).  The inferiority of the BRW for disassortative graphs (here, $\alpha = -0.5$) may be on account of the fact that BRWs on such graphs are more likely to spend much of the search time trapped in a local minimum.  
\begin{figure}
\centering
\includegraphics[width=1.7in]{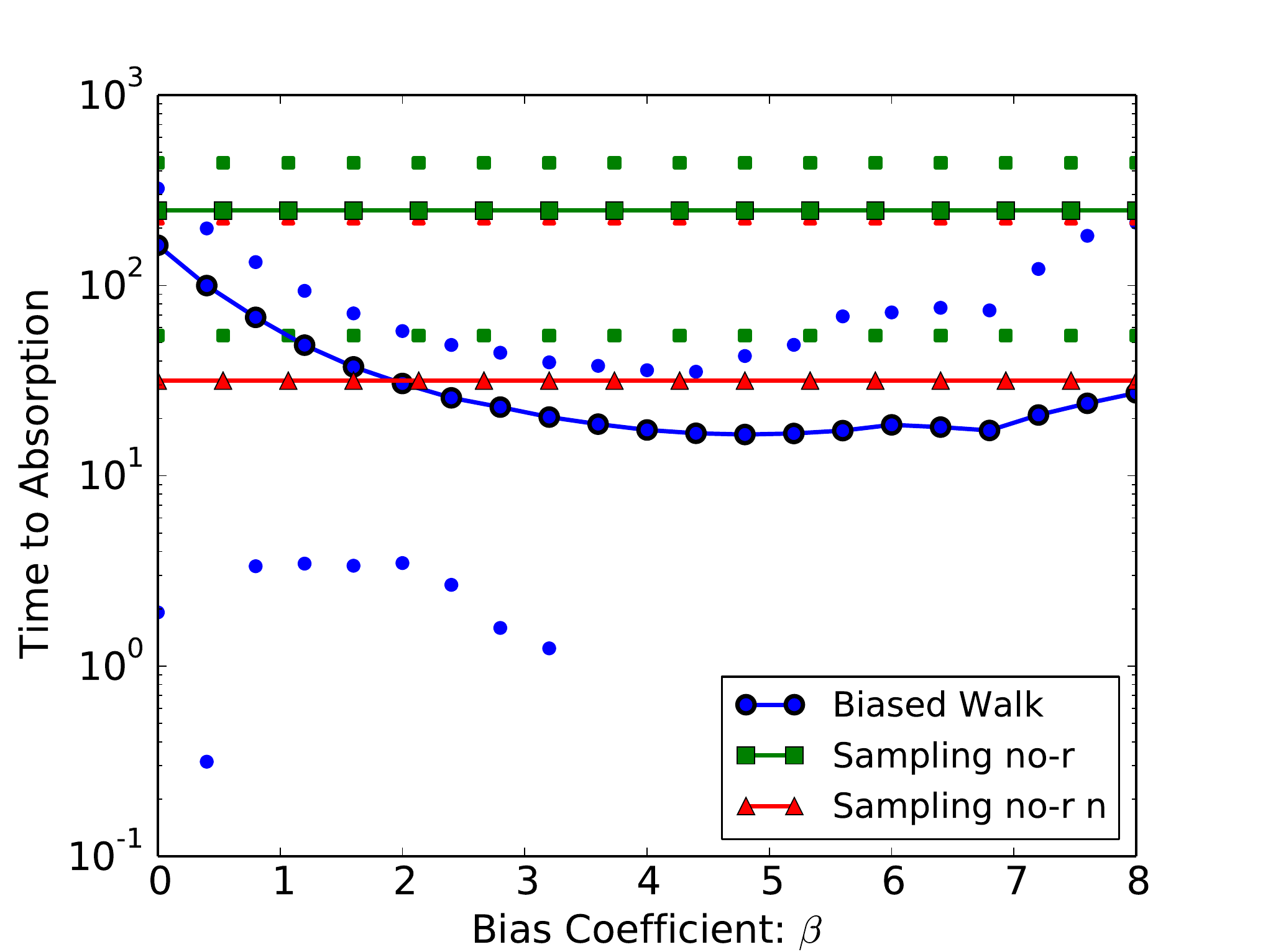}
\includegraphics[width=1.7in]{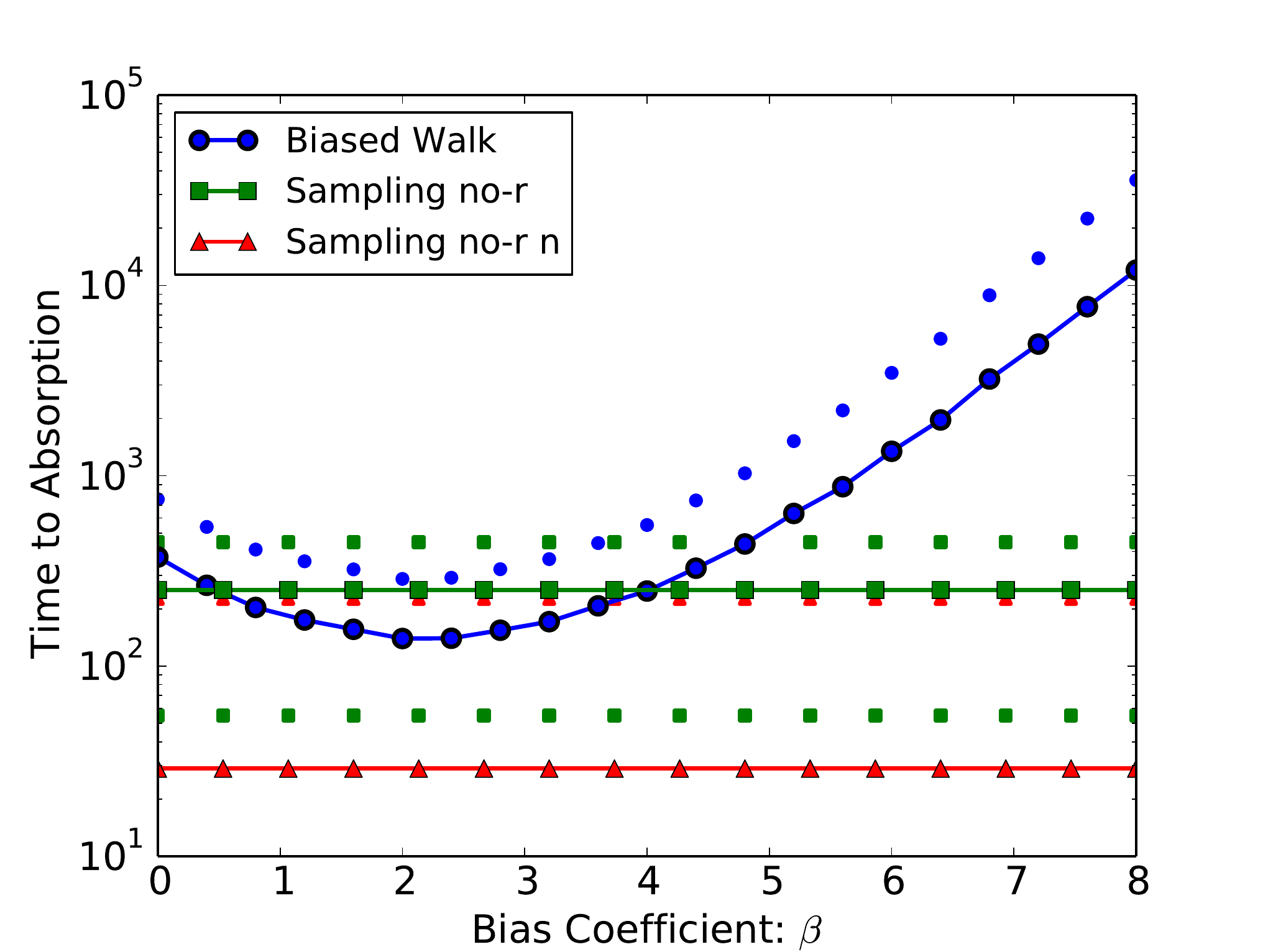}
\caption{ER graph with $\sim 1000$ nodes.  Expected time to absorbtion $\Ebb[T]$ (solid) and $\Ebb[T] \pm \mathrm{Std}(T)$ (dotted) for $i)$ the BRW (blue) and $ii)$ random sampling without replacement (observing $a)$ just the degree of the sampled node (green) and $b)$ degrees of node and its neighbors (red)) versus $\beta$.  $\alpha_t = +0.5$ (left) and $\alpha_t = -0.5$ (right).}
\label{fig:ws1000assdis}
\end{figure}
\begin{figure}
\centering
\includegraphics[width=1.7in]{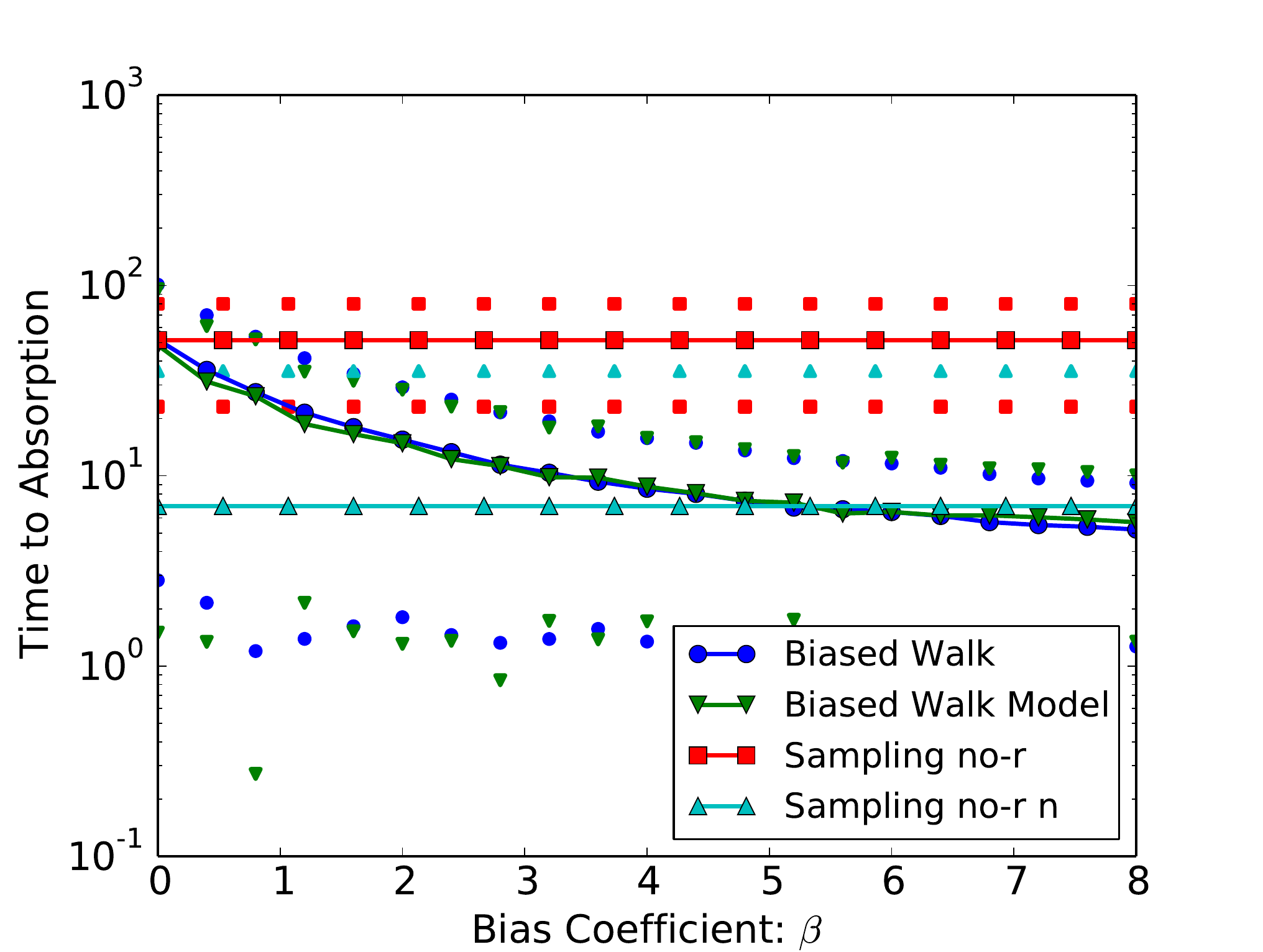}
\includegraphics[width=1.7in]{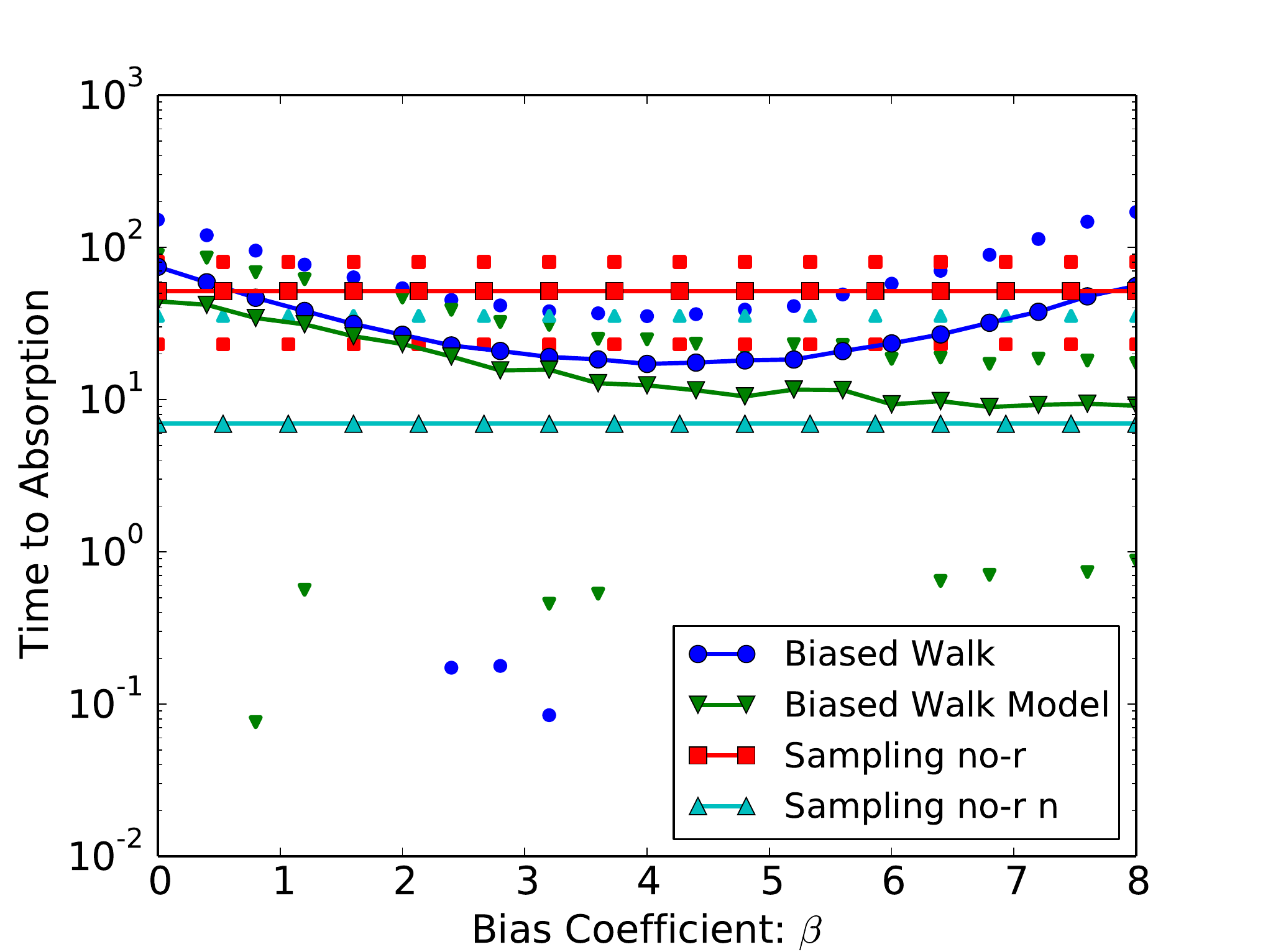}
\caption{Same caption as \figref{ws1000assdis} but for an ER graph with $\sim 100$ nodes.  Also shown is mean absorbtion time $\Ebb[T_W]$ predicted by the model.}
\label{fig:wms100assdis}
\end{figure}
\subsection{Optimal Bias Coefficient $\beta$} 

In this subsection we extend $\alpha_t$ to include the nine values mentioned earlier.  For each $\alpha_t$ we find $\beta^*(\alpha_t)$ and $\Ebb[T^*(\alpha_t)]$ using $\beta^*(\alpha_t)$, and plot both these functions against $\alpha_t$.  Because we observe that the dependence of $\Ebb[T]$ on $\beta$ can be somewhat flat near the optimal, meaning there is some degree of insensitivity to the precise value of $\beta$, we actually compute the interval $[\beta_{\rm min}(\alpha_t),\beta_{\rm max}(\alpha_t)]$ containing $\alpha^*(\alpha_t)$, where the interval holds all values of $\beta$ for which the corresponding value of $\Ebb[T]$ is within $10\%$ of the optimal value $\Ebb[T^*]$.  

The results are shown in \figref{optb} and \figref{optat}.  Several points bear mention.  First, \figref{optb} shows that the optimal bias coefficient $\beta^*$ for $\beta \in [0, \ldots, 8]$ tends to increase with increasing assortativity $\alpha_t$ of the graph.  More sample graphs for each $\alpha_t$, more points $\alpha_t \in [-1,+1]$, and a larger search range for $\beta$ than the current $[0,8]$ are required to confirm this initial observation.  Second, \figref{optat} shows that for both $n \approx 100$ and $n \approx 1000$ there exists an interval of $\alpha$ over which optimized BRWs outperform random sampling of node and neighbor degrees.  Again, more extensive simulations are required.  However, these preliminary results suggest BRWs are inferior to sampling for graphs with negative assortativity.  

\begin{figure}
\centering
\includegraphics[width=1.7in]{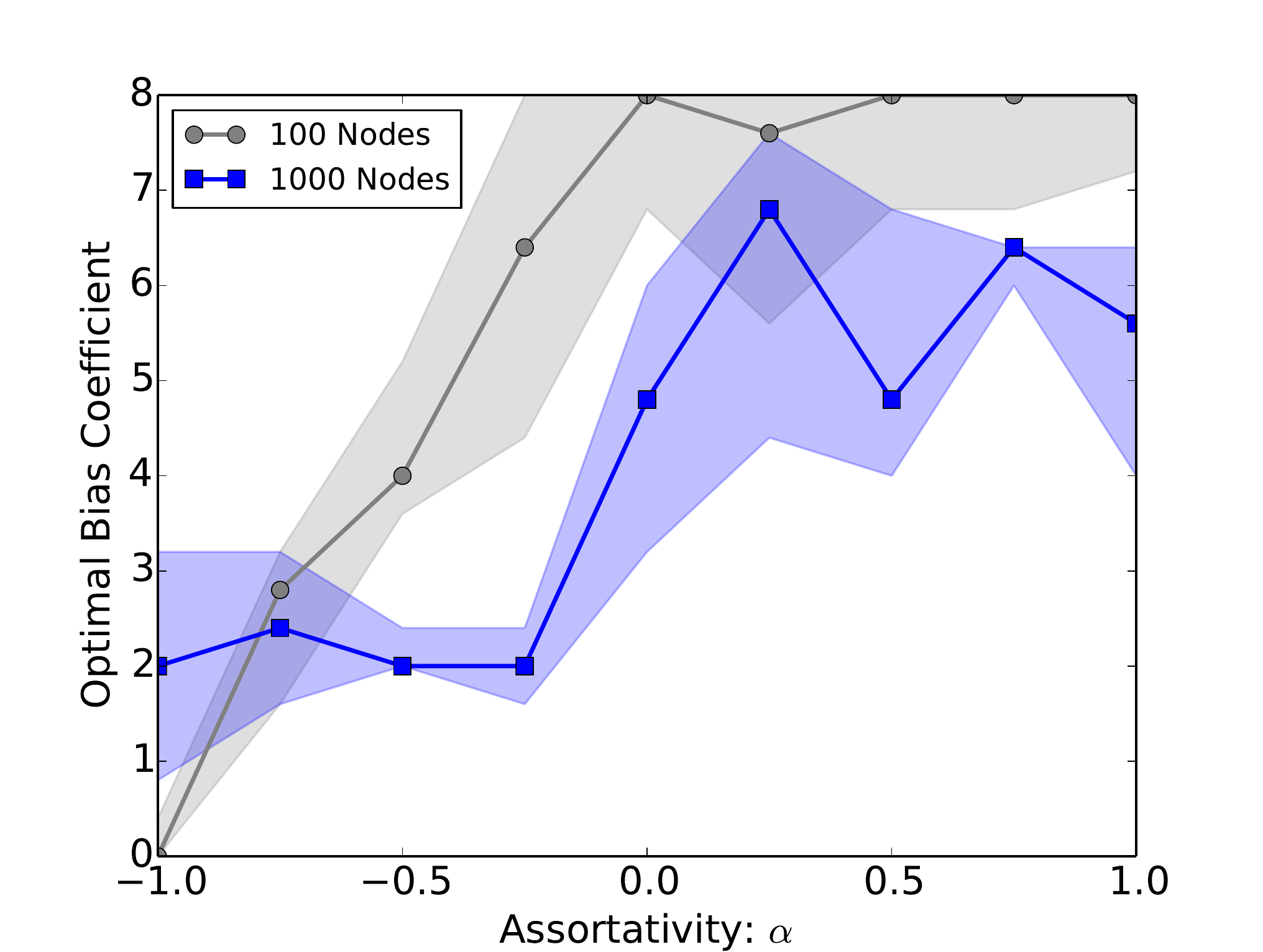}
\includegraphics[width=1.7in]{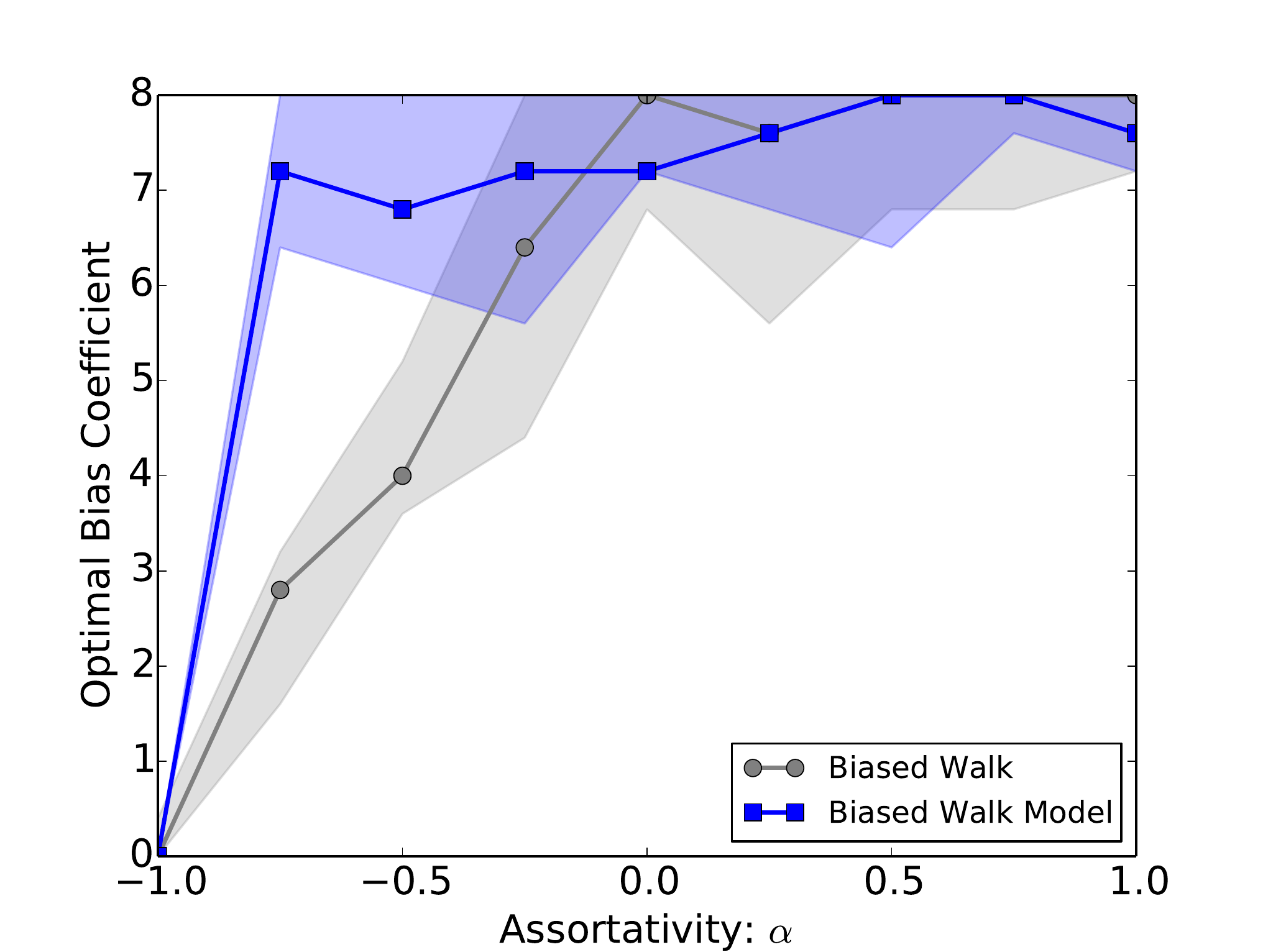}
\caption{The optimal bias coefficient $\beta^*$ (points) and the interval $[\beta_{\rm min},\beta_{\rm max}]$ of points for which $\Ebb[T]$ is within $10\%$ of $\Ebb[T^*]$ (shaded) vs.\ the target assortativity $\alpha_t$.  {\bf Left:} comparison of $n \approx 100$ (gray) and $n \approx 1000$ (blue).  {\bf Right:} comparison of $\beta^*$ for $n=100$ for actual graph (gray) and reduced state space model (blue).}
\label{fig:optb}
\end{figure}

\begin{figure}
\centering
\includegraphics[width=1.7in]{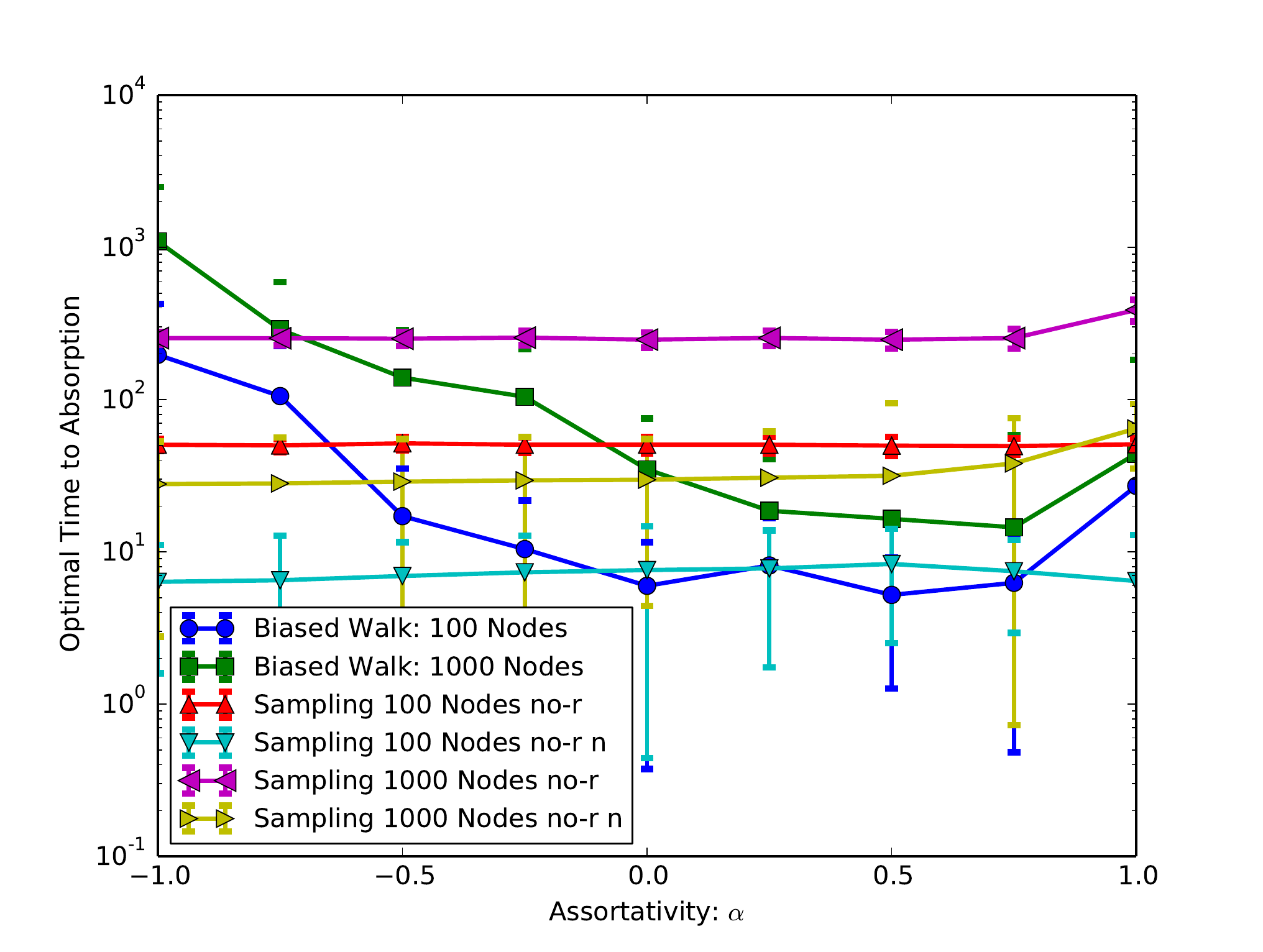}
\includegraphics[width=1.7in]{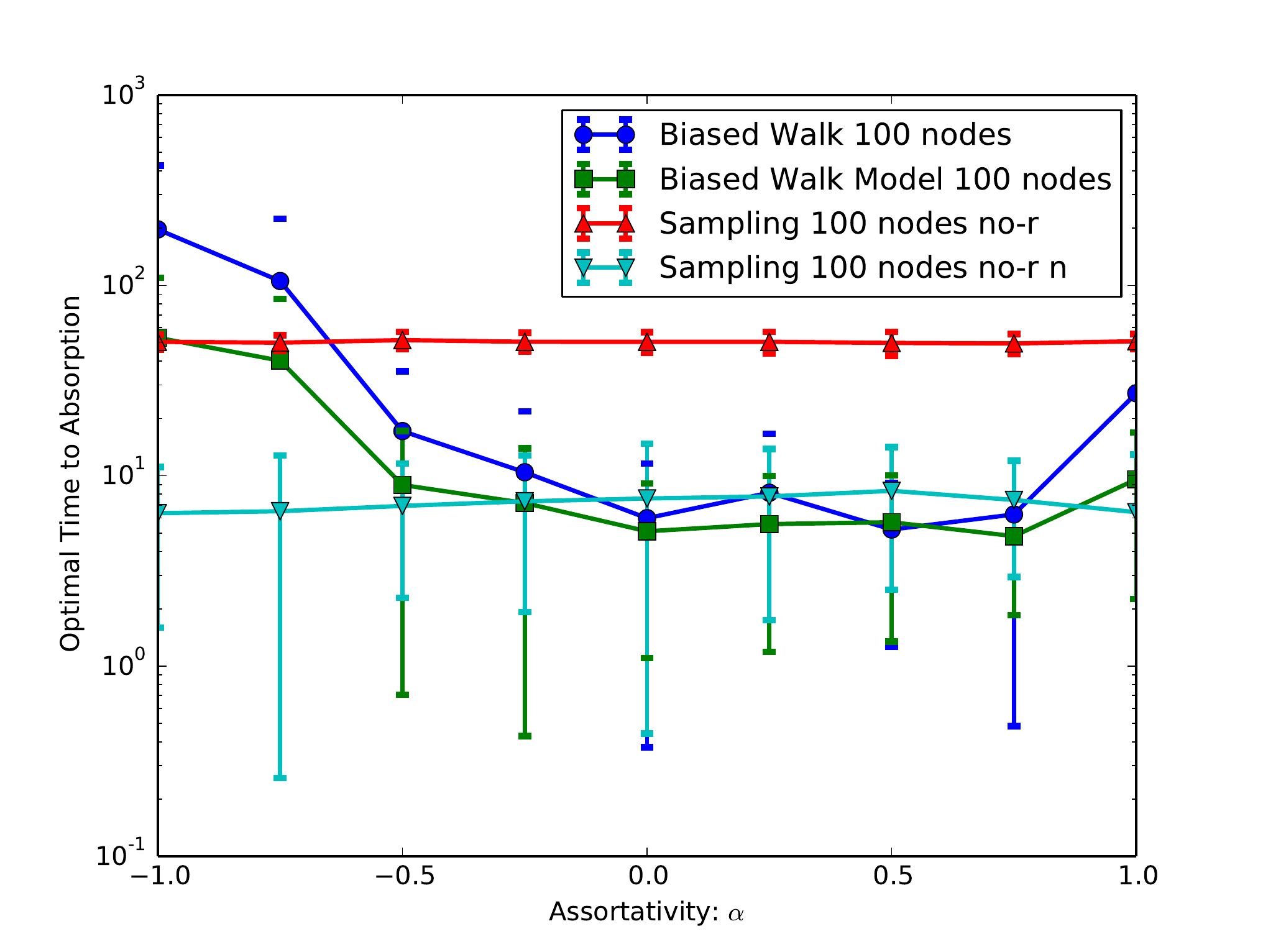}
\caption{Mean absorbtion times $\Ebb[T^*]$ vs.\ $\alpha_t$, where for each $\alpha_t$ we use the optimized value $\beta^*(\alpha_t)$ from \figref{optb}.  {\bf Left:} comparison of $n \approx 100$ (blue) and $n \approx 1000$ (gray).  {\bf Right:} comparison of $\Ebb[T^*]$ for actual graph (blue) and reduced state space model (green).}
\label{fig:optat}
\end{figure}

\subsection{Biased Random Walk Model} 

Finally, the two right side plots in \figref{optb} and \figref{optat} include results for both the simulations of BRWs on the graph as well as analytical computations using the reduced state space model.  In particular, \figref{optb} shows that $\beta^*$ for minimizing $\Ebb[T_V]$ and $\beta^*$ for minimizing $\Ebb[T_W]$ are not exactly equal, but are comparable, and show the same rough increasing trend as a function of $\alpha$.  Moreover, \figref{optat} shows $\Ebb[T_V^*]$ is comparable to $\Ebb[T_W^*]$ for   certain values of $\alpha$.    


\section{Future Work}


The main contributions are $i)$ a potentially useful model for analytically computing the expected absorbtion time of a biased random walk using a reduced state space model, and $ii)$ preliminary comparison of the absorbtion time between a BRW and random sampling to find a max degree node.  Our current goal is to find an approximation for $\Ebb[T_W]$.

\bibliographystyle{IEEEtran}
\bibliography{sources}

\end{document}